\documentclass[11pt]{article}

\usepackage{wrapfig}
\usepackage{graphicx}
\usepackage{palatino}
\usepackage{amsmath}               
\usepackage{amsfonts}              
\usepackage{amsthm}                
\usepackage{amssymb}
\usepackage{todonotes}
\usepackage[pagebackref]{hyperref}
\usepackage[margin=1in]{geometry}
\usepackage{todonotes}

\newtheorem{thm}{Theorem}[section]
\newtheorem{lem}[thm]{Lemma}

\newtheorem{cor}[thm]{Corollary}

\newtheorem{fact}[thm]{Fact}

\newcommand{\dn}{d^{\mathsf{near}}}
\newcommand{\df}{d^{\mathsf{far}}}
\newcommand{\OPT}{\mathsf{OPT}}      
      
\newcommand{\vol}{\mathsf{vol}}

\newcommand{\LP}{\mathsf{FRAC}}

\newcommand{\pvi}[1]{$#1$-Vertex Separator }

\newcommand{\psv}{$k$-Subset Vertex Separator }
\newcommand{\pv}{$k$-Vertex Separator }
\newcommand{\pe}{$k$-Edge Separator }
\newcommand{\pp}{$k$-Vertex Separator }
\newcommand{\pt}{$k$-Path Transversal }

\newcommand{\pvii}[1]{$#1$-Vertex Separator}
\newcommand{\ptii}[1]{$#1$-Path Transversal}

\newcommand{\psvv}{$k$-Subset Vertex Separator}
\newcommand{\pvv}{$k$-Vertex Separator}
\newcommand{\pee}{$k$-Edge Separator}
\newcommand{\ppp}{$k$-Vertex Separator}
\newcommand{\ptt}{$k$-Path Transversal}

\newcommand{\NN}{\mathbb{N}}      

\newcommand{\calP}{\mathcal{P}}

\begin{document}

\setcounter{page}{0}

\title{{\bf Partitioning a Graph into Small Pieces \\ with Applications to Path Transversal}}

\author{
Euiwoong Lee\thanks{Supported by the Samsung Scholarship, the Simons Award for Graduate Students in TCS, and Venkat Guruswami's NSF CCF-1115525. Part of this work was done when the author was an intern at Microsoft Research. {\tt euiwoonl@cs.cmu.edu} }}

\date{Computer Science Department \\ Carnegie Mellon University \\ Pittsburgh, PA 15213.}

\maketitle
\thispagestyle{empty}

\begin{abstract}
Given a graph $G = (V, E)$ and an integer $k \in \NN$, we study {\em \ppp} (resp. {\em \pee}), where the goal is to remove the minimum number of vertices (resp. edges) such that each connected component in the resulting graph has at most $k$ vertices. Our primary focus is on the case where $k$ is either a constant or a slowly growing function of $n$ (e.g. $O(\log n)$ or $n^{o(1)}$). Our problems can be interpreted as a special case of three general classes of problems that have been studied separately (balanced graph partitioning, Hypergraph Vertex Cover (HVC), and fixed parameter tractability (FPT)).

Our main result is an $O(\log k)$-approximation algorithm for \pp that runs in time $2^{O(k)} n^{O(1)}$, and an $O(\log k)$-approximation algorithm for \pe that runs in time $n^{O(1)}$.
Our result on \pe improves the best previous graph partitioning algorithm~\cite{KNS09} for small $k$.
Our result on \pv improves the simple $(k+1)$-approximation from HVC~\cite{BMN15}. When $\OPT > k$, the running time $2^{O(k)} n^{O(1)}$ is faster than the lower bound $k^{\Omega(\OPT)} n^{\Omega(1)}$ for exact algorithms assuming the Exponential Time Hypothesis~\cite{DDvH14}. 
While the running time of $2^{O(k)} n^{O(1)}$ for \pv seems unsatisfactory, we show that the superpolynomial dependence on $k$ may be needed to achieve a polylogarithmic approximation ratio, based on hardness of {\em Densest $k$-Subgraph}. 

We also study $k$-Path Transversal, where the goal is to remove the minimum number of vertices such that there is no simple path of length $k$. 
With additional ideas from FPT algorithms and graph theory, we present an $O(\log k)$-approximation algorithm for $k$-Path Transversal that runs in time $2^{O(k^3 \log k)} n^{O(1)}$. 
Previously, the existence of even $(1 - \delta)k$-approximation algorithm for fixed $\delta > 0$ was open~\cite{Camby15}.
\end{abstract}

\newpage

\section {Introduction}
We study the following natural graph partitioning problems.
\begin{itemize}
\item[] {\bf \pv}
\item[]{\bf Input}: An undirected graph $G = (V, E)$ and $k \in \NN$. 
\item[] {\bf Output}: Subset $S \subseteq V$ such that in the subgraph induced on $V \setminus S$ (denoted by $G|_{V \setminus S}$), each connected component has at most $k$ vertices. 
\item[] {\bf Goal}: Minimize $|S|$. 
\end{itemize}
The edge version can be defined similarly. 

\begin{itemize}
\item[] {\bf \pe}
\item[]{\bf Input}: An undirected graph $G = (V, E)$ and $k \in \NN$. 
\item[] {\bf Output}: Subset $S \subseteq E$ such that in the subgraph $(V, E \setminus S$), each connected component has at most $k$ vertices. 
\item[] {\bf Goal}: Minimize $|S|$. 
\end{itemize}

These two problems have been actively studied in a number of different research contexts that have been developed independently. 
We categorize past research into three groups. 


\paragraph{Graph Partitioning.}
Graph partitioning is a general task of removing a small number of edges or vertices to make the resulting graph consist of smaller connected components. 
In this context, the edge versions have received more attention. 

One of the most well-studied formulations is called {\em $l$-Balanced Partitioning}. Given a graph $G = (V, E)$ and $l \in \NN$, the goal is to remove the smallest number of edges so that the resulting graph has $l$ ($l \geq 2$) connected components with (roughly) the same number $\frac{n}{l}$ of vertices.
\footnote{In the literature it is called $k$-Balanced Partitioning. We use $l$ in order to avoid confusion between $l$-Balanced Partitioning and $k$-Edge Separator ($l = \frac{n}{k}$).}
The case $l = 2$ has been studied extensively and produced elegant approximation algorithms. The best results are 
$O(\log n)$-true approximation (i.e., each component must have $\frac{n}{2}$ vertices)~\cite{Racke08} and $O(\sqrt{\log n})$-bicriteria approximation (i.e., each component must have at most $\frac{2n}{3}$ vertices)~\cite{ARV09}.
The extension to $l \geq 3$ has been studied more recently. 
While it is NP-hard to achieve any nontrivial true approximation for general $l$~\cite{AR06}, Krauthgamer et al.~\cite{KNS09} presented an $O(\sqrt{\log n \log l})$-bicriteria approximation where the resulting graph is guaranteed to have each connected component with at most $\frac{2n}{l}$ vertices. 

The true approximation for $l$-Balanced Partitioning is ruled out by encoding the Integer $3$-Partition problem in graphs, and hard instances contain disjoint cliques of size at most $\frac{n}{l}$. 
Even et al.~\cite{ENRS99} defined a similar problem called {\em $\rho$-Separator}, which is exactly our \pe with $\rho = \frac{k}{n}$. They ``believe that the definition of $\rho$-Separator captures type of partitioning that is actually required in applications'', since ``instead of limiting the number of resulting parts, which is not always important for divide-and-conquer applications or for parallelism, it limits only the sizes or weights of each part.''
They provided a bicriteria approximation algorithm that removes at most $O(\frac{1 + \epsilon}{\epsilon} \log n) \cdot \OPT$ edges to make sure that each component has size $(1 + \epsilon)\rho n$ for any $\epsilon > 0$, which is improved to 
$O(\frac{1 + \epsilon}{\epsilon} \sqrt{\log (1 / \epsilon \rho) \log n}) \cdot \OPT$ by Krauthgamer et al.~\cite{KNS09}. 
The previous algorithms' primary focus is when $\rho$ is a constant (so that $k = \Omega(n)$), and their performance deteriorates when $k$ is small. 
In particular, when $k = O(n^{1 - \epsilon})$ and $\rho = O(n^{-\epsilon})$ for some $\epsilon > 0$, the best guarantee from the above line of work gives an $O(\log n)$-bicriteria approximation algorithm. 

Some of the ideas can be used for the analogous vertex versions, but they have not received the same amount of attention. 
Often additional algorithmic ideas were required to achieve the same guarantee~\cite{FHL08}, or matching the same guarantee is proved to be NP-hard under some complexity assumptions~\cite{LRV13}. 


\paragraph{Hypergraph Vertex Cover.}
\pv for small values of $k$ has been actively studied. \pvi{1}is the famous {\em Vertex Cover} problem. 
When $k = 2$, Papadimitriou and Yannakakis~\cite{Yannakakis81, PY82} defined the {\em dissociation number} to be $n$ minus the optimum of \pvi{2}in the context of certain constrained spanning tree problems, which have been studied independently from the  graph partitioning literature (see~\cite{ODFGW11} for a survey).

A simple $(k + 1)$-approximation for \pv can be achieved by viewing them as a special case of {\em $(k + 1)$-Hypergraph Vertex Cover ($(k + 1)$-HVC)}. Given a graph $G = (V, E_G)$, we construct a hypergraph $H = (V, E_H)$ where $E_H$ contains every set of $k + 1$ vertices $\{ v_1, \dots, v_{k+1} \} \subseteq V$ that induces a single connected component. 
This reduction is complete and sound because a subset $S \subseteq V$ intersects every hyperedge in $E_H$ if and only if $G|_{V \setminus S}$ has no connected component of size at least $k + 1$. 
Since $(k + 1)$-Hypergraph Vertex Cover admits a trivial $(k + 1)$-approximation (e.g., take any hyperedge $e$ not intersecting $S$ and let $S \leftarrow S \cup e$), we get a $(k + 1)$-true approximation for \ppp. This was observed in the work of Ben-Ameur et al.~\cite{BMN15}. 

Approximating $k$-HVC better than the trivial factor $k$ (resp. $k - 1$) will refute the Unique Games Conjecture (resp. $\mathbf{P} \neq \mathbf{NP}$)~\cite{KR08, DGKR03}, so we cannot hope to be able to get a significantly better algorithm for $k$-HVC. An interesting line of research has tried to find a better approximation algorithm when the hypergraph $H$ is promised to have additional structure. When $H$ is $k$-uniform and $k$-partite, Lov{\'a}sz~\cite{Lovasz75} gave a $\frac{k}{2}$-approximation algorithm that is shown to be tight under the Unique Games Conjecture~\cite{GSS15} (the same work also showed almost tight $\frac{k}{2} - 1 + \frac{1}{2k}$ NP-hardness). 

Given two graphs $G, H$ where $H$ is the {\em pattern graph} with $k$ vertices, Guruswami and Lee~\cite{GL15d} studied the problem of removing the minimum number of vertices from $G$ such that the resulting graph has no copy of $H$ as a subgraph. They showed that if $H$ is $2$-vertex connected, this problem is as hard to approximate as the general $k$-HVC. 

\pp can be regarded as a special case of a more general class of problems where we are given a graph $G$ and a set of pattern graphs $\mathcal{H}$ with $k + 1$ vertices and asked to remove the minimum number of vertices to ensure $G$ does not have any graph in $\mathcal{H}$ as a subgraph (in this case $\mathcal{H}$ is the set of all connected graphs with $k + 1$ vertices).

\paragraph{Fixed Parameter Tractability.}
Given a graph $G$ and an integer $k$, the optimum of \pp  has been known as {\em $k$-Component Order Connectivity} in mathematics. 
We refer to the survey by Gross et al.~\cite{GHIKLSS13} for more background. 

Let $\OPT$ be the optimal value. For small values of $k$ and $\OPT$,
the complexity of exact algorithms has been studied in terms of their fixed parameter tractability (FPT). 
While the trivial algorithm takes $n^{O(\OPT)}$ time to find the exact solution for \ppp, Drange et al.~\cite{DDvH14} presented an exact algorithm that runs in time $k^{O(\OPT)} n$, so the problem is in FPT when parameterized by both $k$ and $\OPT$. 
They complemented their result by showing that the problem is $\mathbf{W}[1]$-hard when parameterized by $\OPT$ or $k$. 
They also showed that any exact algorithm that runs in time $k^{o(\OPT)} n^{O(1)}$ will refute the Exponential Time Hypothesis. 

Given an instance of a problem with a parameter $\kappa$, an approximation algorithm is said to be an FPT $c$-approximation algorithm if it runs in time $f(\kappa) \cdot n^{O(1)}$ for some function $f$ and achieves $c$-approximation. 
See the survey of Marx~\cite{Marx08} and the recent work of Chitnis et al.~\cite{CHK13}. 
For \ppp, the simple $(k+ 1)$-approximation runs in polynomial time regardless of $\OPT$ and $k$, but any exact algorithm requires both $\OPT$ and $k$ to be parameterized. 
It is an interesting question whether significantly improved approximation is possible when only one of them is parameterized. 


\subsection {Our Results and Applications}
\label{subsec:results}
Our main result is the following algorithm for \ppp.
For fixed constants $b, c > 1$, an algorithm for \pp is called an $(b, c)$-bicriteria approximation algorithm if given an instance $G = (V, E)$ and $k \in \NN$, it outputs $S \subseteq V$ such that (1) each connected component of $G|_{S \setminus V}$ has at most $bk$ vertices and (2) $|S|$ is at most $c$ times the optimum of \ppp. 

\begin{thm}
For any $\epsilon \in (0, 1/2)$, there is a polynomial time $(\frac{1}{1 - 2 \epsilon}, O(\frac{\log k}{\epsilon}))$-bicriteria approximation algorithm for \ppp. 
\label{thm:main}
\end{thm}
Setting $\epsilon = \frac{1}{4}$ and running the algorithm yields $S \subseteq V$ with $|S| \leq O(\log k) \cdot \OPT$ such that each component in $G|_{V \setminus S}$ has at most $2k$ vertices. Performing an exhaustive search in each connected component yields the following true approximation algorithm whose running time depends exponentially only on $k$. 

\begin{cor}
There is an $O(\log k)$-approximation algorithm for \pp that runs in time $n^{O(1)} + 2^{O(k)} n$. 
\end{cor}
This gives an FPT approximation algorithm when parameterized by only $k$, and its approximation ratio $O(\log k)$ improves the simple $(k+1)$-approximation from $k$-HVC. 
When $\OPT \gg k$, it runs even faster than the time lower bound $k^{\Omega(\OPT)}n^{\Omega(1)}$ for the exact algorithm assuming the Exponential Time Hypothesis~\cite{DDvH14}. 

The natural question is whether superpolynomial dependence on $k$ is necessary to achieve {\em true} $O(\log k)$-approximation. 
The following theorem proves hardness of \pp based on Densest $k$-Subgraph. In particular, a polynomial time $O(\log k)$-approximation algorithm for \pp will imply $O(\log^2 n)$-approximation algorithm for Densest $k$-Subgraph. Given that the best approximation algorithm achieves $\approx O(n^{1/4})$-approximation~\cite{BCCFV10} and $n^{\Omega(1)}$-rounds of the Sum-of-Squares hierarchy have a gap at least $n^{\Omega(1)}$~\cite{BCVGZ12}, such a result seems unlikely or will be considered as a breakthrough. 

\begin{thm}
If there is a polynomial time $f$-approximation algorithm for \ppp, then there is a polynomial time $2f^2$-approximation algorithm for Densest $k$-Subgraph. 
\label{thm:hardness}
\end{thm}

For \pee, we prove that the true $O(\log k)$-approximation can be achieved in polynomial time. This shows a stark difference between the vertex version and the edge version. 
\begin{thm}
\label{thm:edge}
There is an $O(\log k)$-approximation algorithm for \pe that runs in time $n^{O(1)}$. 
\end{thm}
When $k = n^{o(1)}$ so that $\rho = n^{-(1 - o(1))}$, our algorithm outperforms the previous best approximation algorithm for $\rho$-separator~\cite{KNS09, ENRS99}.\footnote{Both papers only present a bicriteria approximation algorithm, but they can be combined with our final {\em cleanup} step to achieve true approximation by {\em adding} $O(\log k)$ to the approximation ratio. See Appendix~\ref{appendix:edge}.} 

While most of graph partitioning algorithms deal with the edge version, we focus on the vertex version because (1) it exhibits richer connections to $k$-HVC and FPT as mentioned, 
(2) usually the vertex version is considered to be harder in the graph partitioning literature. We present the algorithm for the edge version in Appendix~\ref{appendix:edge}.


\paragraph{\ptt.}
Let $l(G)$ be the length of the longest path of $G$ including both endpoints (e.g., length of a single edge is $2$). 
Given a graph $G = (V, E)$ and $k \in \NN$, {\em \ptt} asks to find the smallest subset $S \subseteq V$ such that $l(G|_{V \setminus S}) < k$. 
Finding a path of length $k$ has played a central role in devopment of FPT algorithms --- it is NP-hard to do for general $k$, but there are various algorithms that run in time $2^{O(k)} n^{O(1)}$ using color coding or algebraic algorithms. 

$k$-Path Transversal is motivated by applications in transportation / wireless sensor networks, and has also been actively studied as $k$-Path Vertex Cover or $P_{k}$-Hitting Set~\cite{TZ11, BKKS11, BJKST13, Camby15, Katrenic16} in terms of their approximability and fixed parameter tractability. 
Tu and Zhou~\cite{TZ11} gave a $2$-approximation algorithm for \ptii{3}. Camby~\cite{Camby15} recently gave a $3$-approximation algorithm for \ptii{4}.
In the same doctoral thesis, Camby~\cite{Camby15} asked whether we can get $(1 - \delta)k$-approximation for \pt for a general $k$ and a universal constant $\delta > 0$. We show that it admits $O(\log k)$-approximation in FPT time. Note that the superpolynomial dependence on $k$ is necessary for any approximation from NP-hardness of finding a $k$-path. 

\begin{thm}
There is an $O(\log k)$-approximation algorithm for $k$-Path Transversal that runs in time $2^{O(k^3 \log k)} n^{O(1)}$. 
\label{thm:path}
\end{thm}

\section {Techniques}
\label{sec:techniques}
Our algorithms for \pv and \pe consist of the following three steps. 
We give a simple overview of our techniques for \pvv.

\paragraph{1. Spreading Metrics. }
{\em Spreading metrics} were introduced in Even et al.~\cite{ENRS00} and subsequently used for $\rho$-separator~\cite{ENRS99}.\footnote{The conference version of~\cite{ENRS00} precedes that of ~\cite{ENRS99}.}
They assign lengths to vertices such that any subset $S$ of vertices with $|S| > k$ that induce a single connected component are {\em spread apart}.

Given lengths $x_v$ to each vertex $v \in V$, define $d_{u, v}$ to be the length of the shortest path from $u$ to $v$, including the lengths of both $u$ and $v$ (so that $d_{u, u} = x_u$). Given a feasible solution $S \subseteq V$ for \pvv, let $x_v = 1$ if $v \in S$, and $x_v = 0$ if $v \notin S$. 
It is easy to see that two vertices $u$ and $v$ lie on the same component of $G|_{V \setminus S}$ if and only if $d_{u, v} = 0$. Otherwise, $d_{u, v} \geq 1$. 
Therefore, for every vertex $v$, the number of vertices that have distance strictly less than $1$ from $v$ must be at most $k$. 

Spreading metrics are a continuous relaxation of the above integer program. We relax each distance $x_v$ to have value in $[0, 1]$, and let $d_{u, v}$ still be the length of the shortest path from $u$ to $v$. Let $f_{u, v} = \max(1 - d_{u, v}, 0)$. In the integral solution, it indicates whether $d_{u, v} < 1$ or not. The constraint $\sum_{u} f_{v, u} \leq k$ for all $v \in V$ is a relaxation of the requirement that the number of vertices that have distance strictly less than $1$ from $v$ must be at most $k$. 

Even though this relaxation does not exactly capture the integer problem, one crucial property of this relaxation is that for every $v \in V$ and $\epsilon \in (0, \frac{1}{2})$, the number of vertices that have distance at most $\epsilon$ from $v$ can be at most $\frac{k}{1 - \epsilon}$. This can proved via a simple averaging argument. 

\paragraph{2. Low-Diameter Decomposition. } 
Before we introduce our rounding algorithm, we briefly discuss why the previous algorithms based on the same (or stronger) relaxation has the approximation ratio depending on $n$. 

The current best algorithm by Krauthgamer et al.~\cite{KNS09} further strengthened the above spreading metrics by requiring that they also form an $\ell_2^2$ metric, and transformed them to an $\ell_2$ metric. This black-box transformation of an $n$-points $\ell_2^2$ metric incurs distortion of $\Omega(\sqrt{\log n})$, so the approximation ratio must depend on $n$.

The older work of Even et al.~\cite{ENRS99} used the rounding algorithm of Garg et al.~\cite{GVY96} that iterative takes a ball of small radius from the graph. More specifically, they defined $\vol(v, r)$ to be the total sum of lengths in the ball of radius $r$ centered at $v$, and grow $r$ until the boundary-volume ratio becomes $O(\log (\frac{\vol(v, \frac{1}{2})}{\vol(v, 0)}))$. 
To make $\vol(v, 0)$ nonzero, a {\em seed value} of $\epsilon \cdot \OPT$  must be added to the the definition of $\vol(v, r)$. But when $k = O(1)$ so that the number of balls we need to remove from the graph is $\Omega(n)$, this incurs extra cost of $\Omega(\epsilon n \OPT)$, forcing $\epsilon$ to depend on $n$. 

We apply another standard technique for the {\em low-diameter decomposition} to our spreading metrics. In particular, our algorithm is similar to that of Carlinescu et al.~\cite{CKR05}, preceded by a simple rounding algorithm that removes every vertex with large $x_v$. 
One simple but crucial observation is that the performance of this algorithm only depends on the size of the ball around each vertex, which is exactly what spreading metrics is designed for!
Since the size of each ball of radius $\frac{1}{2}$ is at most $O(k)$, we can guarantee that we can delete at most $O(\log k) \cdot \OPT$ vertices so that each connected component has at most $O(k)$ vertices. 

When $k = O(1)$, to the best of our knowledge, this is a rare example where the number of partitions (i.e., the number of balls taken) is $\Omega(n)$ but the approximation ratio is much smaller than that. The original rounding algorithm of Carlinescu et al.~\cite{CKR05} is applied to {\em $0$-Extension} with $k$ terminals to achieve $O(\log k)$-approximation, where only $k$ balls are needed to be taken. The famous $O(\log k)$-approximation for Multicut with $k$ source-sink pairs~\cite{GVY96} also required only $k$ partitions. 

\paragraph{3. Cleanup.}
After running the bicriteria approximation algorithm to make sure that each connected component has size at most $O(k)$, for \pvv, we run the exhaustive search for each component to have the true approximation. This incurs the extra running time of $2^{O(k)} n$, but our hardness result implies that the superpolynomial dependence on $k$ may be necessary. 

For \pee, essentially the same bicriteria approximation algorithm works. After that, for each component, we use (a variant of) Racke's $O(\log n)$-true approximation algorithm for Min Bisection to each component to make sure that each component has at most $k$ vertices. 
The existence of {\em true approximation} for Min Bisection is a key difference between the vertex version and the edge version. Even $O(\sqrt{\log n})$-bicriteria approximation is known for the vertex version of Min Bisection~\cite{FHL08}, but our hardness result for the vertex version suggests that this algorithm is not likely to be applicable. 
While Min Bisection asks to partition the graph into two pieces while \pe may need to partition it into many pieces, we prove that as long as each connected component has size at most $\frac{3k}{2}$, a simple trick makes the two problems equivalent. 

\section {Algorithm for \pp}
\label{sec:decomposition}

\subsection{Spreading Metrics}
\label{subsec:sa}
Our relaxation is close to {\em spreading metrics} used for $\rho$-separator~\cite{ENRS99}.
While their relaxation involves an exponential number of constraints and is solved by the ellipsoid algorithm, we present a simpler relaxation where the total number of variables and constraints is polynomial. 
Our relaxation has the following variables.

\begin{itemize}
\item $x_v$ for $v \in V$: It indicates whether $v$ is removed or not. 
\item $d_{u, v}$ for $(u, v) \in V \times V$: 
Given $\{ x_v \}_{v \in V}$ as lengths on vertices, 
$d_{u, v}$ is supposed to be the minimum distance between $u$ and $v$. 
Let $\calP_{u, v}$ be the set of simple paths from $u$ to $v$, and given $P = (u_0 := u, u_1, \dots, u_p := v) \in \calP_{u, v}$, let $d(P) = x_{u_0} + \dots + x_{u_p}$. Formally, we want 
\[
d_{u, v} = \min_{P \in \calP_{u, v}} d(P).
\]
Note that $d_{u, v} = d_{v, u}$ and $d_{u, u} = x_u$. 

\item $f_{u, v}$ for all $(u, v) \in V \times V$: 
It indicates whether $u$ and $v$ belong to the same connected component or not.
\end{itemize}

Our LP is written as follows.
\begin{alignat}{3}
\mbox{minimize } \quad & \sum_{v \in V} x_{ v } && \notag \\
\mbox{subject to } \quad & d_{u, v} \leq \min_{P \in \calP_{u, v}} d(P) \qquad &&  \forall (u, v) \in V \times V \label{eq:dist} \\ 
& f_{u, v} \geq 1 - d_{u, v} && \forall (u, v) \in V \times V \notag \\
& f_{u, v} \geq 0 \qquad && \forall (u, v) \in V \times V \notag \\
& \sum_{u \in V} f_{v, u} \leq k && \forall v \in V \label{eq:radius} \\
& x_v \geq 0 && \forall v \in V 
\end{alignat}
\eqref{eq:dist} can be formally written as 
\begin{alignat*}{3}
& d_{u, u} = x_{u} \qquad && \forall u \in V \\
& d_{u, w} \leq d_{u, v} + x_{w} \qquad && \forall (u, v) \in V \times V, (v, w) \in E
\end{alignat*}
Therefore, the size of our LP is polynomial in $n$. It is easy to verify that our LP is a relaxation --- given a subset $S \subseteq V$ such that each connected component of $G|_{V \setminus S}$ has at most $k$ vertices, the following is a feasible solution with $\sum_v x_v = |S|$. 
\begin{itemize}
\item $x_v = 1$ if $v \in S$. $x_v = 0$ if $v \notin S$. 
\item $d_{u, v} = \min_{P \in \calP_{u, v}} d(P)$.
\item $f_{u, v} = 1$ if $u$ and $v$ are in the same component of $G|_{V \setminus S}$. Otherwise $f_{u, v} = 0$. 
\end{itemize}

Fix an optimal solution $\{x_v\}_v, \{ d_{u, v}, f_{u, v} \}_{u, v}$ for the above LP. 
It only ensures that $d_{u, v} \leq \min_{P \in \calP_{u, v}} d(P)$, so a priori $d_{u, v}$ can be strictly less than $\min_{P \in \calP_{u, v}} d(P)$. 
However, in that case increasing $d_{u, v}$ still maintains feasibility, since larger $d_{u,v}$ provides a looser lower bound of $f_{u, v}$ and lower $f_{u, v}$ helps to satisfy~\eqref{eq:radius}. 
For the subsequent sections, we assume that $d_{u, v} = \min_{P \in \calP_{u, v}} d(P)$, and $f_{u, v} = \max(1 - d_{u, v}, 0)$ for all $u, v$. 

\subsection{Low-diameter Decomposition}
\label{subsec:decomposition}
Given the above spreading metrics, we show how to decompose a graph such that each connected component has small number of vertices. 
Our algorithm is based on that of Calinescu et al.~\cite{CKR05}. One major difference is to bound the size of each {\em ball} by $O(k)$ in the analysis, and simple algorithmic steps to ensure this fact.

Fix $\epsilon \in (0, \frac{1}{2})$. 
Given an optimal solution $\{ x_{v} \}_{v \in V}$, the first step of the rounding algorithm is to remove every vertex $v \in V$ with $x_v \geq \epsilon$. 
This simple step is crucial in bounding the size of the ball around each vertex.
It removes at most $\frac{\OPT}{\epsilon}$ vertices.
Let $V' := V \setminus \{ v : x_v \geq \epsilon \}$, and $G' = (V', E')$ be the subgraph of $G$ induced by $V'$. 
Let $d'_{u, v}$ be the minimum distance between $u$ and $v$ in $G'$, and let $f'_{u, v} := \max(1 - d'_{u, v}, 0)$. Since removing vertices only increases distances, $d'_{u, v} \geq d_{u, v}$ and $f'_{u, v} \leq f_{u, v}$ for all $(u, v) \in V' \times V'$.

Our low-diameter decomposition
removes at most $O(\frac{\log k}{\epsilon}) \cdot \sum_{v \in V'} x_v$ vertices so that each resulting connected component has at most $\frac{k}{1 - 2 \epsilon}$ vertices. It proceeds as follows. 
\begin{itemize}
\item Pick $X \in [\epsilon / 2, \epsilon]$ uniformly at random. 
\item Choose a random permutation $\pi : V' \mapsto V'$ uniformly at random.
\item Consider the vertices one by one, in the order given by $\pi$. Let $w$ be the considered vertex (we consider every vertex whether it was previously disconnected, removed or not). 
\begin{itemize}
\item For each vertex $v \in V'$ with $d'_{w, v} - x_v \leq X \leq d'_{w, v}$, 
we remove $v$ when it was neither removed nor {\em disconnected} previously. 
\item The vertices in $\{ v : d'_{w, v} < X \}$ are now disconnected from the rest of the graph. Say these vertices are {\em disconnected}. 
\end{itemize}
\end{itemize}

For each vertex $w$, let $B(w) := \{ v \in V' : d'_{w, v} \leq 2 \epsilon \}$. 
A simple averaging argument bounds $|B(w)|$. 

\begin{lem}
For each vertex $w$, $|B(w)| \leq \frac{k}{1 - 2 \epsilon}$.
\label{lem:ball}
\end{lem}
\begin{proof}
Assume towards contradiction that $|B(w)| > \frac{k}{1 - 2 \epsilon}$. 
For all $u \in B(w)$, 
\[
f_{w, u} \geq f'_{w, u} \geq 1 - d'_{w, u} \geq 1 - 2 \epsilon.
\]
Furthermore, even for $u \notin B(w)$, our LP ensures that $f_{w, u} \geq 0$. 
Therefore, \[
\sum_{u \in V} f_{w, u} \geq 
\sum_{u \in B(w)} f_{w, u} \geq (1 - 2 \epsilon) |B(w)| > k,
\]
contradicting~\eqref{eq:radius} of our LP. 
\end{proof}

Note that at the end of the algorithm, every vertex is removed or disconnected, since every $w \in V'$ becomes removed or disconnected after being considered. Moreover, each connected component is a subset of $\{ v : d'_{w, v} < X \}$ for some $w \in V'$ and $X \leq \epsilon$, which is a subset of $B(w)$. Therefore, each connected component has at most $\frac{k}{1 - 2\epsilon}$ vertices. 
We finally analyze the probability that a vertex $v$ is removed. 
\begin{lem}
The probability that $v \in V'$ is removed is at most $O(\frac{\log k}{\epsilon}) \cdot x_v$. 
\end{lem}
\begin{proof}
Fix a vertex $v \in V'$. 
When $w \in V'$ is considered, $v$ can be possibly removed only if 
\begin{align*}
& d'_{v, w} - x_v \leq \epsilon \\
\Rightarrow \,\, & d'_{v, w} \leq 2\epsilon \qquad \qquad (\mbox{since } x_v \leq \epsilon) \\
\Rightarrow \, \, & w \in B(v).
\end{align*}
Let $W = \{ w_1, \dots, w_p \}$ be such vertices such that $d'_{v, w_1} \leq \dots \leq d'_{v, w_p} \leq 2 \epsilon$. By Lemma~\ref{lem:ball}, $p \leq \frac{k}{1 - 2 \epsilon}$. 

Fix $i$ and consider the event that $v$ is removed when $w_i$ is considered. This happens only if $d'_{v, w_i} - x_v \leq X \leq d'_{v, w_i}$. 
For fixed such $X$, a crucial observation is that if $w_j$ with $j < i$ is considered before $w_i$, since $d'_{v, w_j} - x_v \leq X$, $v$ will be either removed or disconnected when $w_j$ is considered. In particular, $v$ will not be removed by $w_i$. Given these observations, the probability that $v$ is removed is bounded by 

\begin{align*}
\Pr[v \mbox{ is removed}] 
&= \sum_{i = 1}^p \Pr[v \mbox{ is removed when } w_i \mbox { is considered}] \\
&= \sum_{i = 1}^p \Pr[X \in [d'_{v, w_i} - x_v, d'_{v, w_i}] \mbox{ and $w_i$ comes before $w_1, \dots, w_{i - 1}$ in $\pi$}] \\
&\leq \sum_{i = 1}^p \frac{2x_v}{\epsilon i} = x_v \cdot O(\frac{\log p}{\epsilon}) = x_v \cdot O(\frac{\log k}{\epsilon}).
\end{align*}
\end{proof}

Therefore, the low-diameter decomposition removes at most $O(\frac{\log k}{\epsilon}) \cdot \sum_v x_v \leq O(\frac{\log k}{\epsilon}) \cdot \OPT$ vertices so that each resulting connected component has at most $\frac{k}{1 - 2 \epsilon}$ vertices. This gives a bicriteria approximation algorithm that runs in time $poly(n, k)$, proving Theorem~\ref{thm:main}.

\section{$k$-Path Transversal}
\label{sec:path}
Let $G = (V, E)$ and $k \in \NN$ be an instance of \ptt, 
where we want to find the smallest $S \subseteq V$ such that the length of the longest path in $G|_{V \setminus S}$ (denoted by $l(G|_{V \setminus S})$) is strictly less than $k$. 
Recall that the length here denotes the number of vertices in a path. 
Call a path {\em $l$-path} if it has $l$ vertices. 

Let $\calP_k$ be the set of all simple paths of length $k$. 
Our algorithm starts by solving the following naive LP. 
\begin{alignat}{3}
\mbox{minimize } \quad & \sum_{v \in V} x_{ v } && \notag \\
\mbox{subject to } \quad & \sum_{i = 1}^k x_{v_i} \geq 1 \qquad &&  \forall P = (v_1, \dots, v_k) \in \calP_k \label{eq:dist5} \\ 
& x \geq 0 && \forall v \in V \times V \notag
\end{alignat}
When $G$ is a clique with $n$ vertices, any feasible solution needs to remove at least $n - k + 1$ vertices while the above LP has the optimum at most $\frac{n}{k}$ by giving $\frac{1}{k}$ to every $x_v$. Therefore, it has an integrality gap close to $k$, but our algorithm bypasses this gap. 

\begin{lem}
The above LP can be solved in time $k^{O(k)} n^{O(1)}$. 
\end{lem}
\begin{proof}
Given the current solution $\{ x_v \}_v$, we show how to check~\eqref{eq:dist5} in FPT time,
so that the LP can be solved efficiently via the ellipsoid algorithm. 
In particular, it suffices to compute $\min_{P = (v_1, \dots, v_k) \in \calP_k} \sum_{i=1}^k x_{v_i}$. 
Our algorithm is a simple variant of an algorithm for the $k$-Path problem. 
Our presentation follows Williams~\cite{Williams13}. 

Call a set of functions $F = \{ f_i \}_i$ with $f_i : [n] \mapsto [k]$ a {\em $k$-perfect hash family} 
if for any subset $S \subseteq [n]$ with $|S| = k$, there exists $f_i \in F$ such that $f_i(S) = [k]$. 
Naor et al.~\cite{NSS95} show that efficiently computable such $F$ exists with $|F| = 2^{O(k)} \log n$. 

For each $f_i \in F$ and a permutation $\pi \in \mathcal{S}_k$, we construct a directed acyclic graph (DAG) $D_{f_i, \pi}$, 
where for each edge $(u, v) \in E$, we add an arc from $u$ to $v$ if $\pi(f_i(u)) < \pi(f_i(v))$.
Finding the $k$-directed path that minimizes $\sum_{i=1}^k x_{v_i}$ in a DAG can be done via dynamic programming.
For $v \in V$ and $l \in [k]$, let $T[v, l]$ be the minimum weighted length of $l$-path that ends at $v$, and compute $T$
in topological order. 

Let $P^* = (v^*_1, \dots, v^*_k)$ be the path that minimizes $\min_{P = (v_1, \dots, v_k) \in \calP} \sum_{i=1}^k x_{v_i}$.
There must be $f_i \in F$ and $\pi \in \mathcal{S}_k$ such that 
$\pi(f_i(v^*_{i})) < \pi(f_i(v^*_{i + 1}))$ and arc $(v^*_i, v^*_{i+1})$ exists for $1 \leq i < k$. 
For this $f_i$ and $\pi$, the above dynamic programming algorithm for $D_{f_i, \pi}$ finds $P^*$. 

The dynamic programming takes $n^{O(1)}$ time, and we try $2^{O(k)} k! \log n= k^{O(k)} \log n$ different pairs $(f_i, \pi)$,
so the separation oracle runs in time $k^{O(k)} n^{O(1)}$. Our LP has only $n$ variables, so the total LP runs in time
$k^{O(k)} n^{O(1)}$.
\end{proof}

Solve the above LP to get an optimal solution $\{ x_v \}_{v \in V}$. Let $\LP := \sum_{v} x_v$. 
Call a vertex $v \in V$ {\em red} if $x_v \geq \frac{1}{k}$. Let $R$ be the set of red vertices. 
One simple but crucial observation is that every $k$-path must contain at least one red vertices, since all non-red vertices have $x_v < \frac{1}{k}$. 

Let $S^*$ be the optimal solution of $k$-Path Transversal. Let $V^* := V \setminus S^*, R^* := R \setminus S^*$ and $G^* = G|_{V \setminus S^*}$.
The result for $k$-Path Transversal requires the following lemma. 
\begin{lem}
There exists $S' \subseteq V^*$ with $|S'| \leq \frac{|R^*|}{k}$ vertices so that in the induced subgraph $G^*_{V^* \setminus S'}$, each connected component has at most $k^3$ red vertices. 
\end{lem}
\begin{proof}
We prove the lemma by the following (possibly exponential time) algorithm: 
For each connected component $C$ that has more than $k^3$ red vertices, take an arbitrary longest path, remove all vertices in it (i.e., add them to $S'$) 
and charge its cost to all red vertices in $C$ uniformly. Since the length of any longest path should not exceed $k$ and $C$ has more than $k^3$ red vertices, each red vertex in $C$ gets charged at most $\frac{1}{k^{2}}$ in each iteration. 

We argue that each vertex in $G^*$ is charged at most $k$ times. This is based on the following simple observation.
\begin{fact}
In a connected component $C$, any two longest paths should intersect.
\end{fact}
\begin{proof}
Let $P_1 = (v_1, \dots, v_p)$ and $P_2 = (u_1, \dots, u_p)$ be two vertex-disjoint longest paths in the same connected component. Since they are in the same component, there exist $i, j \in [k]$ and another path $P_3 = (v_i, w_1, \dots, w_{q}, u_j)$ such that $w_1, \dots, w_q$ are disjoint from $v$'s and $u$'s ($q$ may be $0$). By reversing the order of $P_1$ or $P_2$, we can assume that $i, j \geq \frac{p + 1}{2}$. Then $(v_1, \dots, v_i, w_1, \dots, w_q, u_j, \dots, u_1)$ is a path with length at least $p + 1$, contradicting the fact that $P_1$ and $P_2$ are longest paths. 
\end{proof}
Therefore, if we remove one longest path from $C$, whether the remaining graph is still connected or divided into several connected components, the length of the longest path in each resulting connected component should be strictly less than the length of the longest path in $C$. 
Therefore, each vertex in $G^*$ can be charged at most $k$ times, and the total amount of charge is $k \cdot \frac{1}{k^{2}} = \frac{1}{k}$. 
\end{proof}

Consider $S^* \cup S'$. Its size is at most $\OPT + \frac{|R^*|}{k} \leq \OPT + \LP \leq 2 \OPT$, since every red vertex has $x_v \geq \frac{1}{k}$, and each component of $G_{S^* \cup S'}$ has at most $k^3$ red vertices. 
We formally define the following generalization of \pvv.

\begin{itemize}
\item[] {\bf \psv}
\item[]{\bf Input}: An undirected graph $G = (V, E)$, a subset $R \subseteq V$ and $k \in \NN$. 
\item[] {\bf Output}: Subset $S \subseteq V$ such that in the subgraph induced on $V \setminus S$ (denoted by $G|_{V \setminus S}$), each connected component has at most $k$ vertices from $R$. 
\item[] {\bf Goal}: Minimize $|S|$. 
\end{itemize}

Even though it seems a nontrivial generalization of \pvv, the analogous bicriteria approximation algorithm also exists. It is proved in Section~\ref{sec:subset}.
\begin{thm}
For any $\epsilon \in (0, 1/2)$, there is a polynomial time $(\frac{1}{1 - 2 \epsilon}, O(\frac{\log k}{\epsilon}))$-bicriteria approximation algorithm for \psvv. 
\label{thm:subset}
\end{thm}

For \ptt, run the above bicriteria approximation algorithm for \psv with $k \leftarrow k^3$ and $\epsilon \leftarrow \frac{1}{4}$. 
This returns a subset $S \subseteq V$ such that $|S| \leq O(\log k) \cdot \OPT$ and each connected component of $G_{V \setminus S}$ has at most $2k^3$ red vertices.

Now we solve for each connected component $C$. Since every $k$-path has to have at least one red vertex, removing every red vertex destroys every $k$-path. In particular, the optimal solution has at most $2k^3$ vertices in $C$. We run the following simple recursive algorithm.
\begin{itemize}
\item Find a $k$-path $P = (v_1, \dots, v_k)$ if exists. 
\begin{itemize}
\item Otherwise, we found a solution --- compare with the current best one and return. 
\end{itemize}
\item If the depth of the recursion is more than $2k^3$, return. 
\item For each $1 \leq i \leq k$, 
\begin{itemize}
\item Remove $v_i$ from the graph and recurse. 
\end{itemize}
\end{itemize}
Finding a path takes time $2^{O(k)} n^{O(1)}$. In each stage the algorithm makes $k$ branches, but the depth of the recursion is at most $2k^3$ and the algorithm is guaranteed to find the optimal solution. Therefore, it runs in time $2^{O(k)} n^{O(1)} \cdot k^{2k^3} = 2^{O(k^3 \log k)} n^{O(1)}$. This proves Theorem~\ref{thm:path}.

\paragraph{Acknowledgements}
We thank Anupam Gupta, Guru Guruganesh, Venkatesan Guruswami, Konstantin Makarychev, and Tselil Schramm for useful discussions. We are also grateful to an anonymous reviewer for FOCS 2016 to suggest a way to improve our algorithms. 

\bibliographystyle{alpha}
\bibliography{../mybib}

\appendix

\section{Hardness of \pp} 
\label{sec:hardness}
In this section, we prove that an $f$-true approximation algorithm for \pp that runs in time $poly(n, k)$ will result in $2f^2$-approximation algorithm for the Densest $k$-Subgraph problem, proving Theorem~\ref{thm:hardness}. 
In particular, $O(\log k)$-true approximation for \pp in time $poly(n, k)$ will lead to $O(\log^2 n)$-approximation for Densest $k$-Subgraph. 

Given an undirected graph $G = (V, E)$ and an integer $k$, Densest $k$-Subgraph asks to find $S \subseteq V$ with $|S| = k$ to maximize the number of edges of $G|_S$. 
It is one of the notorious problems in approximation algorithms. 
The current best approximation algorithm achieves $\approx O(n^{1/4})$-approximation~\cite{BCCFV10}.
While only PTAS is ruled out assuming $\mathbf{NP} \not\subseteq \cap_{\epsilon > 0} \mathbf{BPTIME}(2^{n^{\epsilon}})$~\cite{Khot06}, 
there are strong gap instances for Sum-of-Squares hierarchies of convex relaxations ($n^{\Omega(1)}$ gap for $n^{\Omega(1)}$ rounds)~\cite{BCVGZ12}, so having a $polylog(n)$-approximation algorithm for Densest $k$-Subgraph seems unlikely or will lead to a breakthrough.
Therefore, it may be the case that achieving $O(\log k)$-approximation for \pp requires superpolynomial dependence on $k$ in the running time. 

Our reduction is close to that of Drange et al.~\cite{DDvH14} who reduced Clique to \pp to prove $\mathbf{W}[1]$-hardness. 
Formally, we introduce another problem called {\em Minimum $k$-Edge Coverage}. Given an undirected graph $G$ and an integer $k$, the problem asks to find the minimum number of vertices whose induced subgraph has at least $k$ edges. 
This problem can be thought as a {\em dual} of Densest $k$-Subgraph in a sense that given the same input graph, the optimum of Densest $a$-Subgraph is at least $b$ if and only if the optimum of Minimum $b$-Edge Coverage is at most $a$. 
Hajiaghayi and Jain~\cite{HJ06} proved the following theorem, relating their approximation ratios. 

\begin{thm}[\cite{HJ06}]
If there is a polynomial time $f$-approximation algorithm for Minimum $k$-Edge Coverage, then there is a polynomial time $2f^2$-approximation algorithm for Densest $k$-Subgraph. 
\end{thm}

We introduce a reduction from Minimum $k$-Edge Coverage to \ppp. Given an instance $G = (V, E)$ and $k$ for Minimum $k$-Edge Coverage, the instance of \pp $G' = (V', E')$ and $k'$ is created as follows. Let $n = |V|$, $m = |E|$, and $M = n + 1$. 
\begin{itemize}
\item $V' = V \cup \{ e_{i} : e \in E, i \in [M] \}$. Note that $|V'| = n + Mm$. 
\item $E' = \binom{V}{2} \cup \{ (u, e_i) : u \in V, e \in E, u \in e, i \in [M] \}$. Intuitively, the subgraph induced by $V \subseteq V'$ forms a clique, and for each $e = (u, v) \in E$ and $i \in [M]$, $e_i$ is connected to $u$ and $v$ in $G'$. 
\item $k' = |V'| - Mk$. 
\end{itemize}

\begin{lem}
Every instance of \pp produced by the above reduction has an optimal solution $S \subseteq V'$ such that indeed $S \subseteq V$. 
\label{lem:hardness}
\end{lem}
\begin{proof}
Take an optimal solution $S$ such that $G'_{V' \setminus S}$ has each connected component with at most $k$ vertices. 
Suppose $S$ contains $e_i$ for some $e = (u, v) \in E$ and $i \in [M]$. There are three cases.
\begin{itemize}
\item $u, v \notin S$: Since there is an edge $(u, v) \in E'$, $u$ and $v$ are in the same connected component in $G'|_{V' \setminus S}$. Removing $e_i$ from $S$ and adding $u$ to $S$ still results in an optimal solution. 
\item $u \in S$, $v \notin S$: Removing $e_i$ from $S$ and adding $u$ to $S$ decreases the size of the connected component of $u$ by $1$, and creates a new singleton component consisting $e_i$. It is still an optimal solution. 
\item $u, v \in S$: Removing $e_i$ from $S$ just creates a new singleton component consisting $e_i$. It is a strictly better solution. 
\end{itemize}
We can repeatedly apply one of these three operations until $S$ is an optimal solution contained in $V$. 
\end{proof}

When $S \subseteq V$, $G'|_{V' \setminus S}$ has the following connected components.
\begin{itemize}
\item One component $(V \setminus S) \cup \{ e_i : e = (u, v) \in E, \{ u, v \} \not\subseteq S, i \in [M] \}$. Call it the {\em giant component}. 
\item For each $e = (u, v) \in E$ with $u, v \in S$ and $i \in [M]$, a singleton component $\{ e_i \}$. Call them {\em singleton components}. 
\end{itemize}

Suppose that the instance of Minimum $k$-Edge Coverage admits a solution $T \subseteq V$ such that the induced subgraph $G|_T$ has $l \geq k$ edges. 
Let $S = T$. Since $|V \setminus S| = n - |T|$ and $|\{ (u, v) \in E : \{ u, v \} \not\subseteq T \}| = m - l$, 
In $G'|_{V' \setminus S}$, the giant component will have cardinality
\[
n - |T| + M(m - l) \leq n - |T| + M(m - k) \leq n + M(m - k) = |V'| - Mk = k'.
\]
 
On the other hand, suppose that the instance of \pp has a solution $S$. 
By Lemma~\ref{lem:hardness}, assume that $S \subseteq V$. 
Let $l$ be the number of edges in $G|_S$. 
The size of the giant component is at least $n - |S| + M(m - l) \geq M(m - l - 1) + 1$ since $M > n$. 
Since $S$ is a feasible solution of the \pvii{k'}, we must have 
\begin{align*}
& M(m - l - 1) + 1 \leq k' = M - mk \\
\Rightarrow \,\, & l \geq k. 
\end{align*}
Therefore, $S$ is also a solution to Minimum $k$-Edge Coverage. This proves that the above reduction is an approximation preserving reduction from Minimum $k$-Edge Coverage to \ppp, proving Theorem~\ref{thm:hardness}.

\section {Algorithm for \pe}
\label{appendix:edge}

We present an $O(\log k)$-true approximation algorithm for \pee, proving Theorem~\ref{thm:edge}. Except the cleanup step, the algorithm is almost identical to that of \pvv. 

\subsection{Spreading Metrics}
\label{subsec:sa_edge}

Our relaxation for the edge version is very close to that of the vertex version.
It has the following variables.

\begin{itemize}
\item $x_e$ for $e \in E$: It indicates whether $e$ is removed or not. 
\item $d_{u, v}$ for $(u, v) \in V \times V$: 
Given $\{ x_e \}_{e \in E}$ as lengths on vertices, 
$d_{u, v}$ is supposed to be the minimum distance between $u$ and $v$. 
Let $\calP_{u, v}$ be the set of simple paths from $u$ to $v$, and given $P = (u_0 := u, u_1, \dots, u_p := v) \in \calP_{u, v}$, let $d(P) = x_{u_0, u_1} + \dots + x_{u_{p - 1}, u_p}$. Formally, we want 
\[
d_{u, v} = \min_{P \in \calP_{u, v}} d(P).
\]
Note that $d_{u, v} = d_{v, u}$ and $d_{u, u} = 0$. 

\item $f_{u, v}$ for all $(u, v) \in V \times V$: 
It indicates whether $u$ and $v$ belong to the same connected component or not.
\end{itemize}

Our LP is written as follows.
\begin{alignat}{3}
\mbox{minimize } \quad & \sum_{e \in E} x_{ e } && \notag \\
\mbox{subject to } \quad & d_{u, v} \leq \min_{P \in \calP_{u, v}} d(P) \qquad &&  \forall (u, v) \in V \times V \label{eq:dist_2} \\ 
& f_{u, v} \geq 1 - d_{u, v} && \forall (u, v) \in V \times V \notag \\
& f_{u, v} \geq 0 \qquad && \forall (u, v) \in V \times V \notag \\
& \sum_{u \in V} f_{v, u} \leq k && \forall v \in V \label{eq:radius_2} \\
& x_e \geq 0 && \forall e \in E
\end{alignat}
\eqref{eq:dist_2} can be formally written as
\begin{alignat*}{3}
& d_{u, u} = 0 \qquad && \forall u \in V \\
& d_{u, w} \leq d_{u, v} + x_{v, w} \qquad && \forall (u, v) \in V \times V, (v, w) \in E
\end{alignat*}
Therefore, the size of our LP is polynomial in $n$. It is easy to verify that our LP is a relaxation --- given a subset $S \subseteq E$ such that each connected component of $(V, E \setminus S)$ has at most $k$ vertices, the following is a feasible solution with $\sum_e x_e = |S|$. 
\begin{itemize}
\item $x_e = 1$ if $e \in S$. $x_e = 0$ if $v \notin S$. 
\item $d_{u, v} = \min_{P \in \calP_{u, v}} d(P)$.
\item $f_{u, v} = 1$ if $u$ and $v$ are in the same component of $(V, E \setminus S)$. Otherwise $f_{u, v} = 0$. 
\end{itemize}

Fix an optimal solution $\{x_e\}_e, \{ d_{u, v}, f_{u, v} \}_{u, v}$ for the above LP. 
It only ensures that $d_{u, v} \leq \min_{P \in \calP_{u, v}} d(P)$, so a priori $d_{u, v}$ can be strictly less than $\min_{P \in \calP_{u, v}} d(P)$. 
However, in that case increasing $d_{u, v}$ still maintains feasibility, since larger $d_{u,v}$ provides a looser lower bound of $f_{u, v}$. 
For the subsequent sections, we assume that $d_{u, v} = \min_{P \in \calP_{u, v}} d(P)$ for all $u, v$. 

\subsection{Low-diameter Decomposition}
\label{subsec:decomposition}
Fix $\epsilon \in (0, \frac{1}{2}]$. Given an optimal solution $\{ x_{e} \}_{e \in E}$ and $\{ d_{u, v} \}_{u, v \in V \times V }$ to the above LP, 
our low-diameter decomposition
removes at most $O(\frac{\log k}{\epsilon}) \cdot \sum_e x_e$ edges so that each resulting connected component has at most $\frac{k}{1 - \epsilon}$ vertices. It proceeds as follows. 
\begin{itemize}
\item Pick $X \in [\epsilon / 2, \epsilon]$ uniformly at random. 
\item Choose a random permutation $\pi : V \mapsto V$ uniformly at random.
\item Consider the vertices one by one, in the order given by $\pi$. Let $w$ be the considered vertex (we consider every vertex whether it was previously disconnected or not). 
\begin{itemize}
\item Let $W \leftarrow \emptyset$. 
\item For each vertex $v \in V$ with $d_{w, v} \leq X$, if it is not disconnected yet, add it to $W$. 
\item {\em Disconnect} $W$ from the rest of the graph (i.e., remove every edge that has exactly one endpoint in $W$).
\end{itemize}
\end{itemize}

For each vertex $w$, let $B(w) := \{ v : d_{w, v} \leq \epsilon \}$. 
A simple averaging argument bounds $|B(w)|$. 

\begin{lem}
For each vertex $w$, $|B(w)| \leq \frac{k}{1 - \epsilon}$.
\label{lem:ball_2}
\end{lem}
\begin{proof}
Assume towards contradiction that $|B(w)| > \frac{k}{1 - \epsilon}$. 
For all $u \in B(w)$, $f_{w, u} \geq 1 - d_{w, u} \geq 1 - \epsilon$. 
Furthermore, even for $u \notin B(w)$, our LP ensures that $f_{w, u} \geq 0$. 
Therefore, \[
\sum_{u \in V} f_{w, u} \geq 
\sum_{u \in B(w)} f_{w, u} \geq (1 - \epsilon) |B(w)| > k,
\]
contradicting~\eqref{eq:radius_2} of our LP. 
\end{proof}

Note that at the end of the algorithm, every vertex is disconnected, since every $w \in V$ becomes disconnected after being considered. Moreover, each connected component is a subset of $\{ v : d_{w, v} \leq X \}$ for some $w \in V$ and $X \leq \epsilon$, which is a subset of $B(w)$. Therefore, each connected component has at most $\frac{k}{1 - \epsilon}$ vertices. 
We finally analyze the probability that an edge $e$ is removed. 
\begin{lem}
The probability that $e \in E$ is removed is at most $O(\frac{\log k}{\epsilon}) \cdot x_e$. 
\end{lem}
\begin{proof}
Fix an edge $e = (u, v) \in E$.
For a vertex $v \in W$, let $\dn_{w, e} = \min(d_{w, u}, d_{w, v})$ and $\df_{w, e} = \max(d_{w, u}, d_{w, v})$.
When $w \in V$ is considered, $e$ can be possibly removed only if $\dn_{w, e} \leq \epsilon \Rightarrow w \in B(v) \cup B(u)$. Let $W = \{ w_1, \dots, w_p \}$ be such vertices such that $\dn_{w_1, e} \leq \dots \leq \dn_{w_p, e} \leq \epsilon$. By Lemma~\ref{lem:ball_2}, $p \leq 2 \cdot \frac{k}{1 - \epsilon}$. 

Fix $i$ and consider the event that $e$ is removed when $w_i$ is considered. This happens only if $\dn_{w_i, e} \leq X \leq \df_{w_i, e}$. 
For fixed such $X$, a crucial observation is that if $w_j$ with $j < i$ is considered before $w_i$, since $\dn_{w_j, e} \leq X$, $e$ will be either removed (exactly one of $u$ and $v$ is disconnected) or disconnected (both $u$ and $v$ are disconnected) when $w_j$ is considered. In particular, $e$ will not be removed by $w_i$. Given these observations, the probability that $e$ is removed is bounded by 

\begin{align*}
\Pr[e \mbox{ is removed}] 
&= \sum_{i = 1}^p \Pr[e \mbox{ is removed when } w_i \mbox { is considered}] \\
&= \sum_{i = 1}^p \Pr[X \in [\dn_{v, w_i}, \df_{v, w_i}] \mbox{ and $w_i$ comes before $w_1, \dots, w_{i - 1}$ in $\pi$}] \\
&\leq \sum_{i = 1}^p \frac{2x_e}{\epsilon i} = x_e \cdot O(\frac{\log p}{\epsilon}) = x_v \cdot O(\frac{\log k}{\epsilon}).
\end{align*}
\end{proof}

Therefore, the low-diameter decomposition removes at most $O(\frac{\log k}{\epsilon}) \cdot \sum_v x_v \leq O(\frac{\log k}{\epsilon}) \cdot \OPT$ edges so that each resulting connected component has at most $\frac{k}{1 - \epsilon}$ vertices. 

To get true approximation, we use the algorithm for {\em Balanced $b$-Cut}. 
For an undirected graph $G = (V, E)$ with $n$ vertices and a real $b \in (0, 1/2]$, the Balanced $b$-Cut problem asks to find a subset $S \subseteq V$ with $bn \leq |S| \leq (1 - b)n$ such that the number of edges that have exactly one endpoint in $S$ is minimized. 
Racke~\cite{Racke08} gave an $O(\log n)$-true approximation algorithm for Balanced $b$-cut.\footnote{His algorithm is originally stated for {\em Min Bisection}, the special case with $b = \frac{1}{2}$. For any $c \in [0, 1 - 2b]$, adding a disjoint clique with $cn$ vertices and infinite-weight edges (his algorithm works in weighted version), forces the Minimum Bisection algorithm to output a cut in the original graph where the smaller side contains exactly $\frac{(1 - c)n}{2} \in [bn, \frac{n}{2}]$ vertices. Trying every value of $c \in [0, 1 - 2b]$ that makes $cn$ an integer and taking the best cut gives the desired $O(\log n)$-true approximation for Balanced $b$-Cut. The author thanks to Anupam Gupta for this idea. 
} 

We set $\epsilon = \frac{1}{3}$ such that each connected component after the low-diameter decomposition, each connected component has at most $\frac{3k}{2}$ vertices.
Fix a component of size $k'$. If $k' \leq k$, we are done. Otherwise, we use the $O(\log k') = O(\log k)$-approximation algorithm for Balanced $b$-Cut within the component.
Usually \pe (requires many connected components) and Balanced $b$-Cut (requires $2$ connected components) behave very differently, but given $k' \leq \frac{3k}{2}$, we show that they are equivalent.
\begin{lem}
In a graph $G = (V, E)$ with at most $k' \in (k, \frac{2}{3}k]$ vertices, the optimum solution of \pe and $b$-Balanced Cut with $b = \frac{k' - k}{k'}$ are the same.
\end{lem}
\begin{proof}
Any cut $(S, V \setminus S)$ feasible for $b$-Balanced Cut ensures that $\max(|S|, |V \setminus S|)$ is at most $(1 - b)k' = k$, so it is feasible for \pee. 

For the other direction, given a feasible solution of \pe where $V$ is partitioned into $S_1, \dots, S_l$ (assume $k \geq |S_1| \geq \dots \geq |S_l|$), if $l = 2$, $(S_1, S_2)$ is a feasible solution for $b$-Balanced Cut and we are done. If $l \geq 3$, merge $S_{l - 1}, S_l$ into one set (one $S_i$ may contain multiple connected components). This reduces $l$ by $1$, and since $|S_{l - 1}| + |S_l| \leq \frac{2}{l} \cdot k' \leq \frac{2}{3}k' \leq k$, maintains the invariant that $|S_i| \leq k$ for all $i$. Iterating until $l = 2$ gives a feasible solution for $b$-Balanced Cut with the same number of edges cut. 
\end{proof}

Therefore, running the approximation algorithm $b$-Balanced Cut for each component guarantees that we remove $O(\log k) \cdot \OPT$ additional edges and each component has at most $k$ vertices. This proves Theorem~\ref{thm:edge}.

\section{\psv}
\label{sec:subset}

Given a graph $G = (V, E)$ and $k \in \mathbb{N}$. There is a subset $R \subseteq V$ of {\em red} vertices. Our relaxation has the following variables.

\begin{itemize}
\item $x_v$ for $v \in V$: It indicates whether $v$ is removed or not. 
\item $d_{u, v}$ for $(u, v) \in V \times V$: 
Given $\{ x_v \}_{v \in V}$ as lengths on vertices, 
$d_{u, v}$ is supposed to be the minimum distance between $u$ and $v$. 
Let $\calP_{u, v}$ be the set of simple paths from $u$ to $v$, and given $P = (u_0 := u, u_1, \dots, u_p := v) \in \calP_{u, v}$, let $d(P) = x_{u_0} + \dots + x_{u_p}$. Formally, we want 
\[
d_{u, v} = \min_{P \in \calP_{u, v}} d(P).
\]
Note that $d_{u, v} = d_{v, u}$ and $d_{u, u} = x_u$. 

\item $f_{u, v}$ for all $(u, v) \in V \times V$: 
It indicates whether $u$ and $v$ belong to the same connected component or not.
\end{itemize}

Our LP is written as follows.
\begin{alignat}{3}
\mbox{minimize } \quad & \sum_{v \in V} x_{ v } && \notag \\
\mbox{subject to } \quad & d_{u, v} \leq \min_{P \in \calP_{u, v}} d(P) \qquad &&  \forall (u, v) \in V \times V \label{eq:dist3} \\ 
& f_{u, v} \geq 1 - d_{u, v} && \forall (u, v) \in V \times V \notag \\
& f_{u, v} \geq 0 \qquad && \forall (u, v) \in V \times V \notag \\
& \sum_{u \in R} f_{v, u} \leq k && \forall v \in V \label{eq:radius3}
\end{alignat}
\eqref{eq:dist3} can be formally written as 
\begin{alignat*}{3}
& d_{u, u} = x_{u} \qquad && \forall u \in V \\
& d_{u, w} \leq d_{u, v} + x_{w} \qquad && \forall (u, v) \in V \times V, (v, w) \in E
\end{alignat*}
The only change is that in~\eqref{eq:radius3}, $f_{v, u}$ is summed over $u \in R$ instead of $u \in V$. It is clearly a relaxation. 

Fix an optimal solution $\{x_v\}_v, \{ d_{u, v}, f_{u, v} \}_{u, v}$ for the above LP. 
As usual, assume without loss of generality that $d_{u, v} = \min_{P \in \calP_{u, v}} d(P)$, and $f_{u, v} = \max(1 - d_{u, v}, 0)$ for all $u, v$. 

\subsection{Low-diameter Decomposition}
\label{subsec:decomposition}
Fix $\epsilon \in (0, \frac{1}{2})$. 
Given an optimal solution $\{ x_{v} \}_{v \in V}$, the first step of the rounding algorithm is to remove every vertex $v \in V$ with $x_v \geq \epsilon$. It removes at most $\frac{\OPT}{\epsilon}$ vertices.

Let $V' := V \setminus \{ v : x_v \geq \epsilon \}$, and $G' = (V', E')$ be the subgraph of $G$ induced by $V'$. Let $R' = V' \cap R$. 
Let $d'_{u, v}$ be the minimum distance between $u$ and $v$ in $G'$, and let $f'_{u, v} := \max(1 - d'_{u, v}, 0)$. Since removing vertices only increases distances, $d'_{u, v} \geq d_{u, v}$ and $f'_{u, v} \leq f_{u, v}$ for all $(u, v) \in V' \times V'$.

Our low-diameter decomposition
removes at most $O(\frac{\log k}{\epsilon}) \cdot \sum_{v \in V'} x_v$ vertices so that each resulting connected component has at most $\frac{k}{1 - 2 \epsilon}$  red vertices. It proceeds as follows. 
\begin{itemize}
\item Pick $X \in [\epsilon / 2, \epsilon]$ uniformly at random. 
\item Choose a random permutation $\pi : R' \mapsto R'$ uniformly at random.
\item Consider the red vertices one by one, in the order given by $\pi$. Let $w$ be the considered vertex (we consider every vertex whether it was previously disconnected, removed or not). 
\begin{itemize}
\item For each vertex $v \in V'$ with $d'_{w, v} - x_v \leq X \leq d'_{w, v}$, 
we remove $v$ when it was neither removed nor {\em disconnected} previously. 
\item The vertices in $\{ v : d'_{w, v} < X \}$ are now disconnected from the rest of the graph. Say these vertices are {\em disconnected}. 
\end{itemize}
\end{itemize}

For each vertex $w \in V'$, let $B(w) := \{ v \in R' : d'_{w, v} \leq 2 \epsilon \}$. 
A simple averaging argument bounds $|B(w)|$. 

\begin{lem}
For each vertex $w \in V'$, $|B(w)| \leq \frac{k}{1 - 2 \epsilon}$.
\label{lem:ball3}
\end{lem}
\begin{proof}
Assume towards contradiction that $|B(w)| > \frac{k}{1 - 2 \epsilon}$. 
For all $u \in B(w)$, 
\[
f_{w, u} \geq f'_{w, u} \geq 1 - d'_{w, u} \geq 1 - 2 \epsilon.
\]
Furthermore, even for $u \notin B(w)$, our LP ensures that $f_{w, u} \geq 0$. 
Therefore, \[
\sum_{u \in R} f_{w, u} \geq 
\sum_{u \in B(w)} f_{w, u} \geq (1 - 2 \epsilon) |B(w)| > k,
\]
contradicting~\eqref{eq:radius3} of our LP. 
\end{proof}

Note that at the end of the algorithm, every red vertex is removed or disconnected, since every $w \in V'$ becomes removed or disconnected after being considered. Moreover, each connected component is a subset of $\{ v : d'_{w, v} < X \}$ for some $w \in V'$ and $X \leq \epsilon$, which is a subset of $B(w)$. Therefore, each connected component has at most $\frac{k}{1 - 2 \epsilon}$ red vertices. 
We finally analyze the probability that a vertex $v \in V'$ is removed. 
\begin{lem}
The probability that $v \in V'$ is removed is at most $O(\frac{\log k}{\epsilon}) \cdot x_v$. 
\end{lem}
\begin{proof}
Fix a vertex $v \in V'$. 
When $w \in R'$ is considered, $v$ can be possibly removed only if 
\begin{align*}
& d'_{v, w} - x_v \leq \epsilon \\
\Rightarrow \,\, & d'_{v, w} \leq 2\epsilon \qquad \qquad (\mbox{since } x_v \leq \epsilon) \\
\Rightarrow \, \, & w \in B(v).
\end{align*}
Let $W = \{ w_1, \dots, w_p \}$ be such vertices such that $d'_{v, w_1} \leq \dots \leq d'_{v, w_p} \leq 2 \epsilon$. By Lemma~\ref{lem:ball3}, $p \leq \frac{k}{1 - 2 \epsilon}$. 

Fix $i$ and consider the event that $v$ is removed when $w_i$ is considered. This happens only if $d'_{v, w_i} - x_v \leq X \leq d'_{v, w_i}$. 
For fixed such $X$, a crucial observation is that if $w_j$ with $j < i$ is considered before $w_i$, since $d'_{v, w_j} - x_v \leq X$, $v$ will be either removed or disconnected when $w_j$ is considered. In particular, $v$ will not be removed by $w_i$. Given these observations, the probability that $v$ is removed is bounded by 

\begin{align*}
\Pr[v \mbox{ is removed}] 
&= \sum_{i = 1}^p \Pr[v \mbox{ is removed when } w_i \mbox { is considered}] \\
&= \sum_{i = 1}^p \Pr[X \in [d'_{v, w_i} - x_v, d'_{v, w_i}] \mbox{ and $w_i$ comes before $w_1, \dots, w_{i - 1}$ in $\pi$}] \\
&\leq \sum_{i = 1}^p \frac{2x_v}{\epsilon i} = x_v \cdot O(\frac{\log p}{\epsilon}) = x_v \cdot O(\frac{\log k}{\epsilon}).
\end{align*}
\end{proof}

Therefore, the low-diameter decomposition removes at most $O(\frac{\log k}{\epsilon}) \cdot \sum_v x_v \leq O(\frac{\log k}{\epsilon}) \cdot \OPT$ vertices so that each resulting connected component has at most $\frac{k}{1 - 2 \epsilon}$ red vertices. This gives a bicriteria approximation algorithm that runs in time $poly(n, k)$, proving Theorem~\ref{thm:subset}.

\end{document}